\documentclass[11pt]{article}

\usepackage{fullpage}

\usepackage{amsmath,amsthm, amssymb}
\usepackage{latexsym}
\usepackage{graphicx}
\usepackage{ifthen}
\usepackage{subcaption}
\usepackage{algorithm}
\usepackage[noend]{algpseudocode}
\usepackage[numbers]{natbib}
\usepackage{mathtools}
\usepackage{dsfont}
\usepackage{nicefrac}
\usepackage{wrapfig}
\usepackage{mfirstuc}
\usepackage{soul}

\usepackage[normalem]{ulem}
\usepackage{tikz}
\usetikzlibrary{arrows.meta,chains,decorations.pathreplacing}

\usepackage{hyperref}
\hypersetup{colorlinks=true,linkcolor=blue,citecolor=blue,urlcolor=blue}

\newtheorem*{theorem*}{Theorem}
\newtheorem*{definition*}{Definition}
\newtheorem*{proposition*}{Proposition}
\newtheorem*{corollary*}{Corollary}
\newtheorem{theorem}{Theorem}
\numberwithin{theorem}{section}
\newtheorem{lemma}{Lemma}
\numberwithin{lemma}{section}

\newtheorem{proposition}{Proposition}
\numberwithin{proposition}{section}
\newtheorem{observation}{Observation}
\numberwithin{observation}{section}
\newtheorem{example}{Example}

\usepackage{thm-restate}

\newcommand{\stam}[1]{}

\newcommand{\ronedit}[1]{{#1}}
\newcommand{\ron}[1]{}

\newcommand{\OPT}{\texttt{OPT}}

\newcommand{\E}{\mathbb{E}}

\newcommand{\msi}{\mathcal{S}_i}
\newcommand{\msj}{\mathcal{S}_j}

\newcommand\hmm[1]{\ifnum\ifhmode\spacefactor\else2000\fi>1900 \uppercase{#1}\else#1\fi}

\begin{document}

	\title{On a Competitive Secretary Problem with Deferred Selections}
	\author{%
	Tomer Ezra\thanks{This project has received funding from the European Research Council (ERC) under the European Union's Horizon 2020 research and innovation program (grant agreement No. 866132), and by the Israel Science Foundation (grant number 317/17).
	}\\{\small Tel Aviv University}\\ \\ \texttt{\small tomer.ezra@gmail.com}
	\and	Michal Feldman\footnotemark[1]\\{\small Tel Aviv University} \\ {\small and Microsoft Research}\\ \texttt{\small michal.feldman@cs.tau.ac.il } 
		\and	Ron Kupfer\thanks{This project has received funding from the European Research Council (ERC) under the European Union's Horizon 2020 research and innovation program (grant agreement No. 740282).
		}\\{\small The Hebrew University} \\ {\small of Jerusalem}\\ \texttt{\small ron.kupfer@mail.huji.ac.il}
}
\date{}
	\maketitle	

	\begin{abstract}
		We study secretary problems in settings with multiple agents. 
In the standard secretary problem, a sequence of arbitrary awards arrive online, in a random order, and a single decision maker makes an immediate and irrevocable decision whether to accept each award upon its arrival. 
The requirement to make immediate decisions arises in many cases due to an implicit assumption regarding competition.
Namely, if the decision maker does not take the offered award immediately, it will be taken by someone else.
The novelty in this paper is in introducing a multi-agent model in which the competition is {\em endogenous}.  
In our model, multiple agents compete over the arriving awards, but the decisions need not be immediate; instead, agents may select previous awards as long as they are available (i.e., not taken by another agent). If an award is selected by multiple agents, ties are broken either randomly or according to a global ranking.
This induces a multi-agent game in which the time of selection is not enforced by the rules of the games, rather it is an important component of the agent's strategy.
We study the structure and performance of equilibria in this game. 
For random tie breaking, 
we characterize the equilibria of the  game, and show that the expected social welfare in equilibrium is nearly optimal, despite competition among the agents.
For ranked tie breaking, we give a full characterization of equilibria in the 3-agent game, and show that as the number of agents grows, the winning probability of every agent under non-immediate selections approaches her winning probability under immediate selections.


\end{abstract}

\section{Introduction}
	\label{sec:intro}
	In the classic secretary problem \cite{Ferguson89whosolved} a decision maker observes a sequence of $n$ non-negative real-valued awards $v_1, \ldots, v_n$, which are unknown in advanced, in a random order.
At time step $t$, the decision maker observes award $v_t$, and needs to make an immediate and irrevocable decision whether or not to accept it. If she accepts $v_t$, the game terminates with value $v_t$; otherwise, the award $v_t$ is gone forever and the game continues to the next round.
The objective of the decision maker is to maximize the probability of choosing the maximal award.
A tight competitive ratio of $1/e$ is well known for this problem (see, e.g., \cite{Ferguson89whosolved}).

This problem (and variants thereof) is an abstraction that captures many real-life scenarios, such as an employer who interviews potential workers overtime, renters looking for a potential house, a person looking for a potential partner for life, and so on. 
This problem has interesting implications to mechanism design, auctions and pricing, for both welfare and revenue maximization, in various markets such as online advertising markets (see e.g. \cite{babaioff2008online,BabaioffIK07,EzraFN18,Freeman83Secretary,kesselheim2013optimal,kleinberg2005multiple}).

\paragraph{Competing Agents.}
Most attention in the literature on secretary problems has been given to scenarios in which a single agent makes {\em immediate and irrevocable} decisions. 
The requirement to make an immediate decision arises in many cases from an implicit competition.
Namely, if the decision maker does not take the current offered award, then it may be taken by someone else. 
For example, a potential employee who does not get a job offer following an interview will probably get a job in another firm. 
Indeed, competition among agents is a fundamental component in many real-life online scenarios. 

Recent work has considered the competition aspect in secretary-type problems.
For example, \cite{immorlica2006secretary} and \cite{karlin2015competitive} considered settings with multiple decision makers who compete over awards that arrive online. In these studies, as in the standard setting, decisions are immediate and irrevocable. 

In this paper, we introduce a multi-agent model in which the competition is {\em endogenous}. 
In particular, the different agents compete over the sequence of arriving awards, but unlike previous models, decisions need not be immediate. 
It is the endogenous competition that may drive agents to make fast selections, rather than the rules of the game. 
That is, the time to select an award is part of an agent's strategy.
In particular, every previously arriving award can be selected as long as it has not been taken by a different agent. 
Thus, in our model a sequence of $n$ non-negative real-valued awards $v_1, \ldots, v_n$ arrive over time, unknown from the outset. 
At time step $t$, all agents observe award $v_t$, and need to decide whether to select an available award, or to pass. 

One issue that arises in this model is how to resolve ties among agents. That is, who gets the award if several agents select it. We consider two natural tie-breaking rules; namely, random tie breaking (where ties are broken uniformly at random) and ranked tie-breaking (where agents are a-priori ranked by some global order, and ties are broken in favor of higher ranked agents). Random tie breaking fits scenarios with symmetric agents, whereas ranked tie breaking fits scenarios where some agents are preferred over others, according to some global preference order. For example, it is reasonable to assume that a higher-position/salary job is preferred over lower-position/salary job, or that firms in some industry are globally ordered from most to least desired. 
Random and ranked tie-breaking rules were considered by \citet{immorlica2006secretary} and \citet{karlin2015competitive}, respectively, in secretary settings with immediate and irrevocable decisions.

Two natural objectives have been considered in settings with competition. 
The first is to maximize the probability of receiving the maximal award (see, e.g., \cite{immorlica2006secretary,karlin2015competitive}).
The second is to outperform the competitors (see, e.g., the dueling framework studied by \citet{immorlica2011dueling}). 
We consider an extension of the latter objective, where an agent wishes to maximize the probability to win the $1^{st}$ place, then to win the $2^{nd}$ place, and so on.
Our goal is to study the structure and quality of equilibria in these settings.

%

	\subsection{Our Results and Techniques}
	\label{sec:results}

\subsubsection{Random Tie-Breaking}
For the random tie-breaking rule, we characterize the equilibria of the induced game, and show that the expected social welfare in equilibrium is nearly optimal, despite competition between the agents.
This is cast in Theorems \ref{thm:k-guarantee} and \ref{thm:sec_rand_SW}; a simplified statement follows: 
\begin{theorem*} (Theorems \ref{thm:k-guarantee} and \ref{thm:sec_rand_SW} )
	In every $k$-agent game with random tie-breaking, there exists a simple {\em time-threshold} strategy that guarantees each agent a winning probability of $\frac{1}{k}$, regardless of the strategies of the other agents. 
	The strategy profile where all agents play this strategy is a subgame perfect equilibrium (SPE). 
	Moreover, the expected social welfare of this SPE is at least a $\frac{k-1}{k}$ fraction of the sum of the top $k$ awards.
\end{theorem*}

In particular, we show that each of the $k$ agents can guarantee herself a winning probability of $\frac{1}{k}$ following a simple time-threshold strategy that depends only on the current time, the number of active agents, and whether the maximal award so far is available. 
By symmetry, this is the maximal possible guarantee. 
This guarantee is then used to fully characterize the set of subgame perfect equilibria of the game.

We then establish that in equilibrium, the expected social welfare is at least  $\frac{k-1}{k}$ fraction of the sum of the top $k$ awards which is the optimal welfare (i.e., we bound the {\em price of competition}). We do so by using the following two observations: First, we show that in equilibrium the expected number of selected awards among the top $k$ awards is high. Second, we observe that the probability of an award to be selected in equilibrium is monotone in its rank among the awards. 

We complement this result with a matching upper bound (up to a constant factor), which is derived by observing that in equilibrium there is a constant probability that the first selected award is not one of the top $k$ awards. 

\subsubsection{Ranked Tie-Breaking}
For the ranked tie-breaking rule, we show that for a sufficiently large number of agents, the winning probabilities under immediate- and non-immediate selections are roughly the same.
\begin{theorem*} (Informal  Theorem \ref{thm:sec lex non immidate goes to immediate})
	Under the ranked tie-breaking rule, for every rank $i$, as the number of agents grows, the winning probabilities of the $i^{th}$ ranked agent under non-immediate selections approaches her winning probability under immediate selections. 
\end{theorem*} 

To prove this result, we use observations from \cite{matsui2016lower,karlin2015competitive,EzraFN18} to show that in the immediate decision model, the probability that the maximal award is allocated goes to $1$ as the number of agents grows.
Since an agent in the non-immediate decision model can always mimic the strategy of an agent in the immediate decision model, her winning probability (which equals to her probability of receiving the maximal award) in the non-immediate model is at least her probability of receiving the maximal award in the immediate decisions game. 
We therefore deduce that the winning probabilities in the non-immediate model converge to those in the immediate model.
This claim essentially formalizes the intuition that as competition grows, the urgency to select awards faster grows. 

In addition, we fully characterize the equilibria of the three-agent game. 
\begin{theorem*} (Theorem \ref{thm:three_agents})
	In every equilibrium of the 3-agent game, agent $1$ wins with probability $\frac{4-4\ln2 }{6-4\ln2}\approx0.38$ while each of agents $2,3$ wins with probability $\frac{1}{6-4\ln 2}\approx0.31$.
\end{theorem*}

Notice that agent $1$ (the highest-ranked agent) can always guarantee herself a winning probability of at least $e^{-1}\approx0.37$ by acting according to the optimal strategy in the classical secretary problem. 
The last theorem shows that in a setting with 3 agents, the benefit that agent 1 derives due to her ability to postpone decisions is quite small ($0.38$ vs. $0.37$). 
As implied by Theorem~\ref{thm:sec lex non immidate goes to immediate}, this benefit shrinks as the number of agents grows.

	\subsection{Related Work}
The classical secretary problem and variants thereof have attracted broad interest and have resulted in a vast amount of literature over the years. 
For a comprehensive survey, see, e.g.,~\cite{Freeman83Secretary}.

\paragraph{Competing Agents.} The closest papers to our work are the studies by \citet{karlin2015competitive} and \citet{immorlica2006secretary}, who study  secretary settings with competing agents, with the ranked- and random tie breaking rules, respectively. The main difference between theirs models and ours is that they consider multi-agent settings where agents must make decisions immediately, while in our model the competition is endogenous; namely, past awards can be selected as long as they are available. 
\citet{karlin2015competitive} show that under the ranked tie-breaking rule, the optimal strategy for each agent is a time-threshold strategy, which is given in the form of a recursive formula (albeit not in a closed form). 
\citet{immorlica2006secretary} characterize the Nash equilibria under the random tie-breaking rule.
Another related work is the dueling framework by \citet{immorlica2011dueling}, who considered, among other settings, a dueling scenario between two secretary algorithms, whose objective is to outperform the opponent algorithm.

\paragraph{Matroid- and Uniform Matroid Secretaries.} 
In this paper we derive insights from studies on secretary variants in which a decision maker can choose multiple awards, based on some feasibility constraints. \citet{BabaioffIK07} introduced the matroid secretary problem, where a decision maker selects multiple awards under a matroid constraint. It has been shown that a constant competitive ratio can be achieved for some matroid structures, but the optimal competitive ratio for arbitrary matroids is still open. 
An interesting special case (which was also studied in earlier works such as \cite{kleinberg2005multiple} and \cite{Gilbert66Recognizing}) is one where the decision maker may choose up to $k$ awards (also known as a $k$-uniform matroid constraint). 
\citet{Gilbert66Recognizing,sakaguchi1978dowry,matsui2016lower,EzraFN18} studied secretary models in which a decision maker wishes to maximize the probability of getting the highest award, but may choose up to $k$ awards. We draw interesting connections between these models and the one studied in our paper.

\paragraph{Non-Immediate and Irrevocable Decisions.} Other relaxations of the requirement to select immediately have been considered in the literature.  
\citet{ho2015secretary} consider a {\em sliding-window} variant, where decisions may be delayed for a constant amount of time. A similar model is considered by \citet{KesselheimPV19}, where decisions may be delayed for a randomized (not known in advance) amount of time.
\citet{EzraFN18} study settings where the irrevocability assumption is relaxed. Specifically, they consider a setting where the decision maker can select up to $\ell$ elements immediately and irrevocably, but her performance is measured by the top $k$ elements in the selected set. 
This work is complementary to ours in the sense that it relaxes the irrevocability assumption, while our work relaxes the immediacy assumption.

	\subsection{Paper's Structure}
	Our model is presented in Section \ref{sec:model}.
In Sections \ref{sec:secretary-backward-random} and \ref{sec:secretary-backward-fixed} we present our results with respect to the random tie-breaking rule, and the ranked tie-breaking rule, respectively.
We conclude this paper in Section \ref{sec:future}, where we discuss future directions. 

	\section{Model}
	\label{sec:model}
	\newcommand{\bars}{\bar{S}}
\newcommand{\barp}{\bar{p}}

We consider a variant of the classical secretary setting, where a set of $n$ arbitrary awards are revealed online in a uniformly random order.
Let $v_t$ denote the award revealed at time $t$. 
Unlike the classical secretary problem that involves a single decision maker, in our setting there are $k$ agents who compete over the awards. 
Upon the revelation of award $v_t$, every agent who has not received an award yet may select one of the awards $v_1,\ldots,v_t$ that haven't been assigned yet.
An award that is selected by a single agent is assigned to this agent.
An award that is selected by more than one agent is assigned to one of these agents either randomly (hereafter, {\em random tie breaking}), or according to a predefined ranking (hereafter, {\em ranked tie breaking}). 
Agents who received awards are no longer active. 
Awards that were assigned are no longer available. 
Thus, at time $t$, the set of {\em available} awards is the subset of awards $v_1, \ldots, v_t$ that have not been assigned yet.
The game continues as long as there are active agents. I.e., after time $n$, if active agents remain, the agents compete (without newly arriving awards) on the remaining available awards until all agents are allocated.

Given an instance of a game, the history at time $t$
includes all the relevant information revealed up to time $t$; i.e., $v_1,\ldots,v_t$, and the assignments up to time $t-1$\footnote{In our setting, additional information, such as the history of selections (in contrast to assignments) is irrelevant for future decision making.}.
A strategy of agent $i$, denoted by $S_i$, is a function from the set of all possible histories to a selection decision (either selecting one of the available awards, or passing).
A strategy profile is denoted by $S=(S_1,\ldots,S_k)$.
We also denote a strategy profile by $S=(S_i,S_{-i})$, where $S_{-i}$ denotes the strategy profile of all agents except agent $i$.
Every strategy profile $S$ induces a distribution over assignments of awards to agents. 
For ranked tie breaking, the distribution is with respect to the random order of award arrival, and possibly the randomness in the agent strategies. For random tie breaking, the randomness is also with respect to the randomness in the tie breaking. 

As natural in competition settings, every agent wishes to win the game; that is, to receive an higher award than her competitors. 
We say that agent $i$ wins the $j^{th}$ place in the game if she receives the $j^{th}$ highest award among all allocated awards. 
Let $p=(p_1, \ldots, p_k)$ be a $k$-dimensional vector, where $p_j$ is the probability to win the $j^{th}$ place.
Given two vectors $p,\barp \in R^k$, $p$ is preferred over $\barp$, denoted $\barp \prec p$, if $p$ is lexicographically greater than $\barp$. 
Similarly, $p$ is weakly preferred over $\barp$, denoted $\barp \preceq p$, if $p$ is lexicographically greater or equal to $\barp$. 
That is, every agent wishes to maximize her probability to win the first place, upon equality to maximize the probability to win the second place, and so on.

A strategy profile $S$ induces a $k$-dimensional probability vector $p^i(S)$ for each agent $i$, where $p^i_j(S)$ is the probability that agent $i$ wins the $j^{th}$ place under strategy profile $S$.
Agent $i$ derives higher (respectively, weakly higher) utility from strategy profile $S$ than strategy profile $\bars$, denoted $\bars \prec_i S$ (resp., $\bars \preceq_i S$), if $p^{i}(\bars) \prec p^{i}(S)$ (resp., $p^{i}(\bars) \preceq p^{i}(S)$).
We use $\bars \prec_i S$ and $S \succ_i \bars$ interchangeably, and similarly for $p$ and $\barp$.


Note that $p^{i}_{1}(S)$ is the probability that agent $i$ wins the first place under strategy profile $S$. We sometimes refer to it as agent $i$'s {\em winning probability} under $S$.

\paragraph{Equilibrium notions.}
We consider the following equilibrium notions.
\begin{itemize}
\item Nash Equilibrium: A strategy profile $S=(S_1,\ldots, S_k)$ is a {\em Nash equilibrium} (NE) if for every agent $i$ and every strategy $S'_i$, it holds that $(S'_i,S_{-i}) \preceq_i (S_i,S_{-i})$.

\item Subgame perfect equilibrium: A strategy profile $S=(S_1,\ldots, S_k)$ is a {\em subgame perfect equilibrium} (SPE) if it is a NE for every subgame of the game. I.e. for every initial history $h$, $S$ is a NE in the game induced by history $h$. 
\end{itemize}

SPE is a refinement of NE; namely, every SPE is a NE, but not vice versa.
		


	\section{Random Tie-Breaking}
	\label{sec:secretary-backward-random}
	In this section, we study the setting of the random tie-breaking rule.
We characterize the SPEs and give simple time threshold strategies with optimal utility guarantees. We show that in the SPE where all agents play according to this optimal guarantee strategy, at least $\frac{k-1}{k}$ of the optimal social welfare is achieved in expectation.

Consider the following strategy $\sigma_i$ for agent $i$, where $\ell_t$ denotes the number of active agents at time $t$ (including agent $i$).


\begin{itemize}
	\item If $t \geq n$, then select the maximal available award.
	\item If $\frac{n}{\ell_t} < t <n$, and the maximal award so far is available, then select it.
	\item If $ t = \frac{n}{2}  $, $\ell_t=2$, and the maximal award so far is available, then select it.
	\item Else, pass (i.e., select no award).		
\end{itemize}

We denote by $\msi$ the set of strategies in which agent $i$ plays according to $\sigma_i$ up to the following cases:
\begin{itemize}
	\item If $\ell_t=2$, $t=\frac{n}{2}$, and both the highest and the second highest awards so far are available, then the agent can either pass or select the highest award so far.
	\item If $\ell_t=1$ and the highest award so far is available, then the agent can either pass or select the highest available award.
\end{itemize}

We next show that strategies in $\msi$ are the only strategies that guarantee a utility of at least $(\frac{1}{k},\ldots,\frac{1}{k})$.
By symmetry, there is no strategy $\sigma$ such that $\left(\frac{1}{k},\ldots,\frac{1}{k}\right) \prec p^i(\sigma,S_{-i})$ for every $S_{-i}$. 
Thus, the strategy profiles where each agent $i$ plays according to a strategy in $\msi$, are the only SPEs.
\begin{theorem} \label{thm:k-guarantee}
For every agent $i$, $\sigma \in \msi$ and $S_{-i}$, it holds that $\left(\frac{1}{k},\ldots,\frac{1}{k}\right) \preceq p^i(\sigma,S_{-i})$ . 
For every $\sigma' \notin \msi$ there exists $S_{-i}$ such that $\left(\frac{1}{k},\ldots,\frac{1}{k}\right) \succ p^i(\sigma',S_{-i})$.
Moreover, the strategy profiles where each agent $j$ plays according to a strategy in $\msj$, are the only SPEs.
\end{theorem}
Before proving Theorem~\ref{thm:k-guarantee}, we show the following:
\begin{observation}
For every time $t<n$, selecting an element that is not the maximal so far cannot guarantee a winning probability of $\frac{1}{k}$.
\end{observation}
Thus, we can assume that agents do not select elements that are not the maximal so far up to time $t=n$.
We now give lower bounds on the probabilities of winning first and second places in a strategy $\sigma \in \msi$ given a time $t$, and whether the maximal and second maximal awards so far are available.
Let $A_t^\ell\in[0,1]^2$ be an ordered pair denoting a lower bound on the probabilities of agent $i$ winning first and second places under strategy profile $(\sigma,S_{-i})$, conditioned on the event that at time $t$ (after observing the award $v_t$, but before making selections in time $t$) agent $i$ is active, there are $\ell_t$ active agents (including agent $i$), and the maximal and second maximal awards up to time $t$ are available. 
Similarly, let $B_t^\ell$ be a lower bound on the probabilities of agent $i$ winning first and second places under strategy profile $(\sigma,S_{-i})$, conditioned on the event that at time $t$, agent $i$ is active, there are $\ell_t$ active agents (including agent $i$), and the maximal award up to time $t$ is available, but the second maximal award is not available.

Let $C_t^\ell$ be a lower bound on the probabilities of agent $i$ winning first and second places under strategy profile $(\sigma,S_{-i})$, conditioned on the event that {\em after the allocations} of time $t$, agent $i$ is active, there are $\ell_t$ active agents (including agent $i$), and the maximal award up to time $t$ is not available, but the second maximal award is available.
Finally, let $D_t^\ell$ be a lower bound on the probabilities of agent $i$ winning first and second places under strategy profile $(\sigma,S_{-i})$, conditioned on the event that after the allocations of time $t$, agent $i$ is active, there are $\ell_t$ active agents (including agent $i$), and none of the maximal and the second maximal awards up to time $t$ are available.

In Lemma \ref{lem:atbt} we lower bound the above terms. 
The full proof of the lemma is deferred to the appendix.
\begin{lemma} \label{lem:atbt}
For every $t$, and every $\ell_t>1$, it holds that:
\begin{itemize}
	
	\item $A_t^{\ell_t} \geq \left(\frac{1}{\ell_t},\frac{1}{\ell_t}\right)$
	\item 	$B_t^{\ell_t} \geq \left(\frac{1}{\ell_t},\frac{(n-t)(n+t-1)}{n(n-1)\ell_t}\right)$
	\item $C_t^{\ell_t} \geq \left(\frac{n-t}{n\ell_t},\frac{n^2+t^2-tn-n}{n(n-1)\ell_t}\right)$
	\item 	$D_t^{\ell_t} \geq \left(\frac{n-t}{n\ell_t},\frac{n-t}{n\ell_t}\right)$
\end{itemize}
\end{lemma}

\begin{proof}[Proof sketch]
We observe that the winning probability of an agent depends on the time step $t$, the number of active agents $\ell_t$, and whether the maximal award so far is available or not.
If at time $t$ an agent receives the maximal award up to time $t$, she wins with probability $\frac{t}{n}$ (which is the probability that this award is the global maximum).
If another agent receives the maximal award up to time $t$, then by symmetry, each remaining active agent can guarantee a winning probability of $\frac{n-t}{n(\ell_t-1)}$.
Thus, selecting the maximal award so far is better whenever $\frac{t}{n} > \frac{n-t}{n(\ell_t-1)}$, and passing is better whenever $\frac{t}{n} < \frac{n-t}{n(\ell_t-1)}$.
In cases where $\frac{t}{n} = \frac{n-t}{n(\ell_t-1)}$, both passing and selecting the maximal award so far give a winning probability of $\frac{1}{\ell_t}$, and the agents break this tie based on the probability of winning the second place.
For the four states of whether  the maximal and second maximal awards so far are available, we establish lower bounds on the probabilities of winning first and second places, by induction on $t$ and $\ell_t$.  
\end{proof}
We are now ready to prove Theorem~\ref{thm:k-guarantee}.
\begin{proof}[Proof of Therorem~\ref{thm:k-guarantee}]
It follows from the proof of Lemma~\ref{lem:atbt} that for every strategy that is not in $\msi$, if each agent $j\neq i$ plays according to a strategy in $\msj$,  agent $i$'s utility is smaller than $\left(\frac{1}{k},\frac{1}{k},0,\ldots,0,\frac{k-2}{k}\right)$.
It also shows that if agent $i$ plays according to a strategy in $\msi$, and there exists an agent $j$ that plays according to a strategy not in $\msj$, is greater than $\left(\frac{1}{k},\frac{1}{k},\frac{k-2}{k},0,\ldots,0\right)$.

Thus, the only SPEs are profiles in which each agent $j$ plays according to a strategy in $\msj$, and by symmetry, we get that the utility of agent $i$ is exactly $\left( \frac{1}{k},\ldots,\frac{1}{k}\right)$, as desired.
\end{proof}

We next show that despite the competition, the social welfare of the SPE where all agents play according to $\sigma_i$ is at least $\frac{k-1}{k}$ of the maximal possible social welfare. 
We use $y_i$ to denote the $i^{th}$ maximal value among $\{v_j\}_j$.
The optimal social welfare is $\OPT=\sum_{i=1}^k y_i$.
\begin{theorem}\label{thm:sec_rand_SW}
	The expected sum of the allocated awards in the SPE profile  $S=(\sigma_1,\ldots,\sigma_k)$ is at least $\frac{k-1}{k}\cdot \OPT$.
\end{theorem}
\begin{proof}
	Let $X$ be the random variable denoting the first time in which one of the top $k$ awards appears, let $A$ be the number of awards selected in the profile $S$ before time $X$, and let $B$ be the number of awards $v_t$ for $t\in[\lfloor\frac{n}{k}\rfloor+1,X-1]$ that are the maximal so far upon their arrival. 
	Notice that at most one award is selected up to time $\frac{n}{k}$, and any award that gets selected at time $t>\frac{n}{k}$ is the best award so far, thus, $B+1\geq A$. 
	We start by bounding the expectation of $A$. 
	\begin{eqnarray}
	\nonumber
	\E[A] & = & \E\left[A \mid X>\frac{n}{k}\right] \cdot \Pr\left[X>\frac{n}{k}\right]+
	\E\left[A \mid X\leq \frac{n}{k}\right] \cdot \Pr\left[X \leq \frac{n}{k}\right]\\\label{eq:1}
 & \leq & \E\left[B+1 \mid X>\frac{n}{k}\right] \cdot \Pr\left[X>\frac{n}{k}\right] + 0 \\\nonumber
	&=& 
	 \Pr(X > \frac{n}{k})\cdot 1 + \sum_{x=\lfloor\frac{n}{k}\rfloor+1}^{n}\Pr(X=x)\cdot\sum_{t=\lfloor\frac{n}{k}\rfloor+1}^{x-1}\frac{1}{t}\label{eq:2}\\
	&=& \Pr(X > \frac{n}{k}) + \sum_{t=\lfloor\frac{n}{k}\rfloor+1}^{n}\Pr(X> t)\cdot\frac{1}{t}\\\nonumber
	&\leq &  \left(\frac{n-\frac{n}{k}}{n}\right)^k + \sum_{t=\lfloor\frac{n}{k}\rfloor+1}^{n}\frac{1}{t} \left(\frac{n-t}{n}\right)^k\\\nonumber
	& \leq & \frac{1}{e}+\sum_{t=\lfloor\frac{n}{k}\rfloor+1}^{n}\frac{e^{-\frac{tk}{n}}}{t}  \\\nonumber
	& \leq & \frac{1}{e}+\frac{k}{n}\sum_{t=\lfloor\frac{n}{k}\rfloor+1}^{n}e^{-\frac{tk}{n}}  \\ 
	& \leq & \frac{2.7}{e} <1,\label{eq:3}
	\end{eqnarray}
	where Inequality~\eqref{eq:1} is since if $X\leq \frac{n}{k}$, then $A=0$, and since  $B+1\geq A$. 
	Equality~\eqref{eq:2} is true since $v_t $ is the maximal so far with probability $\frac{1}{t}$.  
	Inequality~\eqref{eq:3} holds for  every $k\leq n$.
	We conclude that in expectation, less than one award among the top $k$ awards is not selected. Since the arrival order is uniform and the algorithm decisions depend only on the current ranks of awards and not on their actual values,
the probability that award $y_j$ (i.e., the $j$th highest award) is chosen is monotonically decreasing in $j$. Monotonicity follows by the fact that given an order of awards such that $y_{j+1}$ is chosen while $y_{j}$ is not, if these two awards are switched, then $y_j$ is chosen while $y_{j+1}$ is not. The monotonicity together with the fact that in expectation at least $k-1$ awards out of the highest $k$ awards are chosen, gives expected social welfare of at least $\frac{k-1}{k}\sum_{i=1}^k{y_i}$. 
\end{proof}

We complement this result by showing an instance of awards $y_1,\ldots y_n$ where the social welfare in every SPE is at most $\frac{k-\Omega(1)}{k}\cdot \OPT$ for every $k>1$.
\begin{example}
	Suppose  $y_1 = \ldots =  y_k = 1$ and $y_j=0$ for every $j$ such that $k<j\leq n$. 
	 In every SPE, the first selection is made at time no later than $t=\lfloor \frac{n}{k}\rfloor + 1$. If none of the top $k$ awards appeared up to time $t$, at least one of the agents gets an award of 0. The probability that none of $y_1,\ldots y_k$ appeared by time $t$ is approximately $(\frac{k-1}{k})^k=\Omega(1)$. Thus, the expected social welfare is at most $\frac{k-\Omega(1)}{k}\cdot \OPT$.
\end{example}
	
	\section{Ranked Tie-Breaking}
	\label{sec:secretary-backward-fixed}

In this section, we study competition under the ranked tie-breaking rule.

We first claim that it is without loss of generality to assume that for every agent, the winning probability equals to the probability of receiving the highest award. To show this, we observe that whenever exactly one agent is active, she may as well wait until time $t=n$ and only then select the maximal award without harming her utility. Thus, it can be assumed that the maximal award is always allocated, and the winning agent receives it. Thus, we may assume that the first-order objective of every agent is to maximize the probability of receiving the maximal award, as in the standard secretary problem and previous multi-agent extensions.

In Section~\ref{subsec:general} we present general observations regarding equilibria in this setting. 
We then characterize the equilibrium in the 3-agent game in Section~\ref{subsec:2-3-agents}.
In Section~\ref{subsec:winning} we show that as the number of competing agents goes to infinity, the agents' probabilities of receiving the highest award (which equal to the agents' winning probabilities) converge to the corresponding probabilities  in the immediate decisions model described by \citet{karlin2015competitive}.


\subsection{General Observations}
\label{subsec:general}

We first make observations about the structure of the subgame perfect equilibria (SPE) of the game.

\begin{proposition}
A strategy profile $S=(S_1,\ldots,S_n)$ is an SPE if for every agent $i$, $S_i$ is described by a set of time thresholds $T_j^\ell$ for every $j,\ell$ such that $1\leq j \leq \ell \leq k$. 
At time $t$, agent $i$ selects the highest award so far if it is available and $t \geq T_j^\ell$, where the current number of active agents is $\ell$ and agent $i$ is ranked $j^{th}$ among them\footnote{If $t = T_j^\ell$, then the agent is indifferent between selecting and passing.}. In addition, if $\ell$ agents are active and $t > n-\ell$, then the lowest-ranked active agent makes a selection, even if the highest award so far is not available.  
\end{proposition}

\begin{proof}
	Since the objective is lexicographic, agents wish to maximize the probability to win the first place, and only if this is hopeless, they will attempt to win lower places. This implies that agents make selections only if at least one of the following conditions holds: 
	\begin{enumerate}
		\item The selected award has a non-zero probability of being the maximal allocated award.
		\item Winning the first place has zero probability, independent of whether a selection is made. 
	\end{enumerate}
In the first case, an agent can either make a selection of the highest award so far or pass. This decision depends only on the number of awards revealed so far and the number of active agents with higher ranks. For a given number  of active agents, the winning probability is monotonically increasing in the number of revealed awards. Thus, there exists some $T_j^\ell$ such that at time $t$, agent $i$ selects the highest award so far if it is available and $t \geq T_j^\ell$, where $i$ is ranked $j$-th among the $\ell$ active agents.
	 
The second case can only occur if the number of active agents exceeds the number of remaining awards.
We show by induction that if there are $\ell$ active agents and $t > n-\ell$, the best strategy of the lowest ranked active agent is to select the best available award. for the base of the induction, if a single agent remains at time $t>n-1$, she clearly makes a selection. For the induction step, suppose $t > n-\ell$. The lowest ranked agent among the $\ell$ active agents knows (by the induction assumption) that at any future time step $t\leq n$, a selection is going to be made by a higher ranked agent. Thus, she clearly makes a selection.
\end{proof}

We proceed with several observations about the time thresholds in the SPE of the game.
 
Since any agent can always mimic the strategy of an agent ranked lower than her, in equilibrium a lower-ranked agent would be willing to receive any award that a higher-ranked agent would be willing to receive.
In the threshold terminology, it means that:
\begin{observation} \label{obs:monotonicity}
For any number of active agents $\ell$, for every pair of ranks $h,j$ such that $h < j \leq \ell$, without loss of generality it holds that $T_j^\ell \leq T_{h}^\ell$.
\end{observation}

The following observation gives bounds on the time threshold of the lowest-ranked active agent relative to the second-lowest ranked active agent. 
\begin{observation}\label{obs:last2threshold}
For any number of active agents $\ell$, it holds that $T_{\ell-1}^\ell \geq T_{\ell}^\ell \geq T_{\ell-1}^\ell -1$.
\end{observation}

\begin{proof}
By Observation \ref{obs:monotonicity} we have that $T_{\ell-1}^\ell \geq T_{\ell}^\ell$.
On the other hand, the lowest-ranked active agent never makes a selection before time $\min_{j\neq\ell} T_j^\ell - 1$, because she can only benefit from waiting as long as no other active agent makes a selection. The claim now follows since, by Observation \ref{obs:monotonicity}, $\min_{j\neq\ell} T_j^\ell = T_{\ell-1}^\ell$.
\end{proof}

Recall that the winning probability of agent $i$ under strategy profile $S$ is denoted by $p_1^i(S)$. 
Throughout this section, we make two simplifications in notation. First, we omit $S$. Second, we omit the subscript 1, since we consider only the probability of winning the $1^{st}$ place. 
Consequently, we denote the probability that agent $i$ wins the $1^{st}$ place in strategy profile $S$ by $p_i$.

The following observation gives bounds on the winning probability of the lowest-ranked agent relative to the second-lowest agent.

\begin{observation}
\label{obs:pik2}
It holds that $p_{k-1}-\frac{1}{n}\leq p_k \leq p_{k-1}$. 
\end{observation}

\begin{proof}
	By definition of $T_i^k$, the $i^{th}$-ranked agent is willing to select the maximal award among $v_1,\ldots,v_{T_i^k}$ and is not willing to select the maximal award among $v_1,\ldots,v_{T_{i}^k-1}$. 
	Thus, 
	\begin{equation}
	\frac{T_i^k-1}{n} \leq p_i \leq \frac{T_i^k}{n}. \label{eq:pik}
	\end{equation}
	When $T_i^k$ is strictly smaller than $T_{i-1}^k$, agent $i$ can guarantee herself a winning probability of $\frac{T_i^k}{n}$ and the right inequality becomes equality. That is, $p_i = \frac{T_i^k}{n}$. 
	Thus, by Observation~\ref{obs:last2threshold}, either (i) ${T_k^k}={T_{k-1}^k}$, in which case  $p_k,p_{k-1} \in \left[\frac{T_k^k-1}{n},\frac{T_k^k}{n}\right]$, or (ii) 
	${T_k^k}={T_{k-1}^k-1}$, in which case $p_k = \frac{T_i^k}{n}$ and $p_k \leq p_{k-1} \leq \frac{T_{k-1}^k}{n} = \frac{T_{k}^k+1}{n} = p_k + \frac{1}{n}$. The observation follows.
\end{proof}


\subsection{The 3-Agent Game}
\label{subsec:2-3-agents}
In a 2-agent game Observation \ref{obs:pik2} implies that $\frac{1}{n}\geq p_1-p_2\geq 0$. That is, both agents win with probability roughly a half.
%
%
This symmetry breaks as more agents join the game and the setting becomes interesting already in the case of 3 agents. 

Notice that the highest-ranked agent can always guarantee herself a probability of at least $e^{-1}\approx0.37$ to receive the highest award by adopting the optimal strategy in the classical secretary problem.
An interesting question is whether the opportunity to make non-immediate decisions increases this probability for the highest-ranked agent. 

We show that in a game with 3 agents this advantage exists, but is very small. That is, the winning probability of the highest-ranked agent is nearly the same as in the immediate decision model. More accurately, we show that:
\begin{theorem}\label{thm:three_agents}
	In a setting with 3 agents, in any SPE, agent $1$ wins with probability $\approx0.38$, while each of agents $2,3$ wins with probability $\approx0.31$.
\end{theorem}
\begin{proof}
We start with the case where $T_1^2=T_2^2=\frac{n}{2}$, and $T_2^3=T_3^3$ and later on show how to handle the cases where the equality is broken. Let $\tau = T_2^3 = T_3^3$.
	At time $\tau$ both agents $2,3$ select the maximal award so far and it is allocated to agent $2$. 
	Hence, agent 2's winning probability is: 
	\begin{equation}
	p_2=\frac{\tau}{n}. \label{eq:p23}
	\end{equation}
	Agent $3$ wins in two cases. Case 1: The maximal award arrives at time $t$  for $\tau < t <\frac{n}{2}$. Case 2: Agent $3$ does not select an award before time $\frac{n}{2}$ (this happens with probability $\frac{\tau}{n/2}$), agent $1$ selects an award $v_s$ at time $s\geq \frac{n}{2}$ (such selection is made if and only if $v_s$ is maximal so far, which happens with probability $\frac{1}{s}$) 
	, and the maximal award arrives later than time $s$ (this happens with probability $1-\frac{s}{n}$). Hence:  
	\begin{equation}
	p_3=\frac{n/2-\tau}{n} + \frac{\tau}{n/2} \cdot\sum_{s=n/2+1}^{n}{\frac{1}{s}\left( 1-\frac{s}{n}\right)} \approx \frac{1}{2} - \frac{\tau}{n}+2\left( \ln 2 - \frac{1}{2}\right)\frac{\tau}{n} . \label{eq:p33}	
	\end{equation}
	Combining the fact that $p_2\approx p_3$ (by Observation~\ref{obs:pik2}) with Equations \eqref{eq:p23} and \eqref{eq:p33} gives:
	\begin{eqnarray*}
		\frac{\tau}{n}  \approx \frac{1}{2} - \frac{\tau}{n}+2\left( \ln 2 - \frac{1}{2}\right)\frac{\tau}{n}.	
	\end{eqnarray*}
	Solving for $\tau$, we get $\tau =\frac{n}{6-\ln 16}\approx 0.3098n$.
	The assertion of the theorem follows.
	
	In the proof we assumed that $T_2^3=T_3^3$.	By Observation~\ref{obs:last2threshold}, the equality is sometimes broken and we have that $T_2^3=T_3^3+1$. In this case, agent $3$ makes the first selection and the "names" of agents $2$ and $3$ are switched in the remainder of the proof. That is, Equation~\ref{eq:p23} and~\ref{eq:p33} denote the winning probabilities of agent $3$ and $2$ respectably. The same comment applies with respect to the two active agents at time $\frac{n}{2}$. 
	
\end{proof}


\subsection{Immediate vs. Non-Immediate Selection Models}
\label{subsec:winning}
In this section, we compare the immediate and non-immediate models for games with a large number of agents.
Let $p_{i,k}$ denote the probability that agent $i$ wins in a $k$-agent game, with non-immediate selections, and let $q_i$ denote the probability that agent $i$ receives the highest award in a game with immediate selections. 
We note that under immediate selections, $q_i$ is independent of the number of agents in the game. 

The main result here is that agents' winning probabilities in equilibrium under non-immediate selections approaches their winning probabilities under immediate selections, as the number of agents grows.

\begin{theorem}\label{thm:sec lex non immidate goes to immediate}
	For every $i$  $$ \lim_{k\rightarrow \infty}p_{i,k} = q_i.$$   
\end{theorem}

Before we prove our main theorem, we restate a result by Karlin and Lei \cite{karlin2015competitive} regarding the immediate decision model.

\begin{theorem}[\citet{karlin2015competitive}]\label{thm:karlin}
	For every $n,k$ and every $i \in [k]$, there is a unique $T_i$ (independent of k) such that agent $i$ plays a $T_i$-threshold strategy in SPE; namely, wait until time $T_i$, then make a selection whenever a best-so-far award appears. It holds that $T_{i-1} \geq T_i$, and $q_i = \frac{T_i}{n}$, for all $i$. 
\end{theorem}

We are now ready to present the proof of Theorem~\ref{thm:sec lex non immidate goes to immediate}.
\begin{proof}
\citet{matsui2016lower} showed an interesting connection between a $k$-agent game with immediate selections, and a scenario where a single decision maker is allowed to select $k$ (out of $n$) awards, and wishes to maximize the probability of getting the maximal award. Denote this last probability by $\tau_k$. Specifically, they show that:
\begin{equation}
\sum_{i=1}^k q_i = \tau_k. \label{eq:qk}
\end{equation} 
It is also known that  (see, e.g., \cite{Gilbert66Recognizing} and  \cite{EzraFN18})
\begin{equation}
\lim_{k\rightarrow\infty}\tau_k=1. \label{eq:qk2}
\end{equation}

By combining Equations \eqref{eq:qk} and \eqref{eq:qk2}, it follows that:
\begin{equation}\label{eq:pi1}
\lim_{k\rightarrow\infty}{\sum_{i\leq k}{q_i}}=1.
\end{equation}

A straightforward corollary of Theorem \ref{thm:karlin} is that every agent $i \in [k]$ has a strategy $S_i$ that guarantees her a winning probability of at least $q_i$ in under non-immediate selections, independent of the strategies played by other agents. To see this, observe that the $T_i$-threshold strategy gives this guarantee, by the monotonicity of $T_i$. Formally, for every number of agents $k$ and every $i \in [k]$, 
\begin{equation}\label{eq:upper}
p_{i,k} \geq  q_i.
\end{equation}

Combining all of the above we get that:
\begin{eqnarray}
\lim_{k\rightarrow \infty} p_{i,k} & {=}  & \lim_{k\rightarrow \infty} (1- \sum_{j\in [k]\setminus\{i\}} p_{j,k}) \nonumber \\
& \stackrel{\eqref{eq:upper}}{\leq} & \lim_{k\rightarrow \infty}(1- \sum_{j\in [k]\setminus\{i\}} q_j )\nonumber \\
&=&\lim_{k\rightarrow \infty}(1+q_i- \sum_{j\in [k]} q_j  )\nonumber \\
&\stackrel{\eqref{eq:pi1}}{=}& q_i\label{eq:pkpi}.
\end{eqnarray}
The assertion of the theorem follows by Equations  \eqref{eq:upper} and \eqref{eq:pkpi}. 	
\end{proof}

%

	\section{Discussion and Future Directions}
	\label{sec:future}

In this work we study secretary settings with competing decision makers. 
While in previous secretary settings, including ones where competition among multiple agents is considered, decisions must be made immediately, we introduce a model where the time of selection is part of the agent's strategy, and thus the competition is endogenous. In particular, decisions need not be immediate, and agents may select previous awards as long as they are available.
We believe that this setting captures many real-world settings, where agents compete over ``awards" that may remain available until taken by a competitor.  

This work suggests open problems and directions for future research.
For the ranked tie-breaking rule, we fully characterize the equilibria of a 3-agent game, and derive the corresponding utilities of the agents. Extending this characterization to any number of agents is an interesting open problem. 

Below we list some future directions that we find particularly natural. 

\begin{itemize}
\item Study competition in additional problems related to optimal stopping theory, such as prophet and pandora box settings.
\item Study competition in secretary settings under additional tie-breaking rules, such as random tie breaking with non-uniform distribution, and tie-breaking rules that allow to split awards among agents.
\item Study competition in secretary settings under additional feasibility constraints. For example, scenarios where agents can choose up to $k$ awards, or other matroid constraints. 
\item Extend the current study to additional objective functions. 
\end{itemize}


	\bibliographystyle{ACM-Reference-Format}
	\bibliography{secretaries}
	\appendix	
	\section{Proof of Lemma~\ref{lem:atbt}}
	We prove the lemma by induction on $t$ and $\ell_t$ ($t=n,\ldots,1$, and $\ell_t=1,\ldots,k$). 
	For the base of the induction (either $t=n$ or $\ell_t=1$), observe that:
\begin{itemize}
	\item 	$A_t^1 =B_t^1 = (1,0)$ for all $t$. Indeed, if at some point, agent $i$ is the only active agent, and the highest award so far is still available, then agent $i$ selects the maximal award. Thus, the winning probability is $1$.
	
	\item $A_n^{\ell_t} \geq (\frac{1}{\ell_t},\frac{1}{\ell_t})$ for all $\ell_t$.
	Indeed, at time $n$ according to strategy $\sigma$, agent $i$ selects the maximal award, and therefore wins with probability at least $\frac{1}{\ell_t}$. If he looses, then he selects the maximal award and consequently wins the second place with probability at least $\frac{\ell_t-1}{\ell_t}\cdot \frac{1}{\ell_t-1}= \frac{1}{\ell_t}$.
	
	\item $B_n^{\ell_t} \geq (\frac{1}{\ell_t},0)$ for all $\ell_t$.
	Indeed, at time $n$ according to strategy $\sigma$, agent $i$ selects the maximal award, and therefore wins with probability at least $\frac{1}{\ell_t}$.
	
	\item	$C_n^{\ell_t} \geq (0,\frac{1}{\ell_t})$ for all $\ell_t$. Indeed, if the highest award (among all the $n$ awards) has been already allocated, then the agent selects the maximal available award, which is the second highest, and if he gets it, then he wins second place. The probability of winning second place is at least $\frac{1}{\ell_t}$.
	
	\item $C_t^1=(\frac{n-t}{n},\frac{t}{n})$ for every $t$.
	According to $\sigma$, agent $i$ waits until the maximal award so far is available, and if such an award does not arrive, then he selects the maximal award available at time $n$. If the highest award appears between time $t+1$ and $n$, then agent $i$ wins first place; else, he wins second place.
	
	\item	$D_n^{\ell_t} \geq (0,0)$ for all $\ell_t$. This holds trivially.
	\item $D_t^1=\left(\frac{n-t}{n},\frac{t(n-t)}{n(n-1)}\right)$ for every $t$.
	According to $\sigma$, agent $i$ waits until the maximal award so far is available, and if such an award does not arrive, then he selects the maximal award available at time $n$. If the highest award appears between time $t+1$ and $n$, then the agent wins first place (which happens with probability $\frac{n-t}{n}$). Else, if the second maximal award appears at this time (which happens with probability $\frac{t(n-t)}{n(n-1)}$), he wins second place.

	\end{itemize}
\ronedit{Notice that the case of $\ell_t=1$ has different bounds than for general $\ell_t$.}
We now show the step of the induction. Let $t<n$, and $\ell_t>1$. The following hold:
\begin{enumerate}
\item For the case where $\ell_t=2$ and $\frac{t}{n} \geq \frac{1}{\ell_t}$, according to $\sigma$, agent $i$ competes over the maximal award so far. Assuming that $\ell' \in \{1,2\}$ agents (including $i$) compete over this award, we get that:
\begin{eqnarray*}
	A_t^{\ell_t} &\geq&  \min_{1\leq \ell'\leq \ell_t }\left\{\frac{1}{\ell'} \cdot \left(\frac{t}{n},\frac{n-t}{n}\right) + \frac{\ell'-1}{\ell'}C_t^{\ell_t-1} \right\}\\
	&\geq&  \min_{1\leq \ell'\leq \ell_t }\left\{\frac{1}{\ell'} \cdot \left(\frac{t}{n},\frac{n-t}{n}\right) + \frac{\ell'-1}{\ell'}\left(\frac{n-t}{n(\ell_t-1)},\frac{t}{n}\right)\right\}\\  & =  &
	\frac{1}{2} \cdot \left(\frac{t}{n},\frac{n-t}{n}\right) + \frac{1}{2}\left(\frac{n-t}{n},\frac{t}{n}\right)\\  & = &  \left(\frac{1}{2},\frac{1}{2}\right),
\end{eqnarray*} where the first inequality holds since the probability of receiving the award under competition is $\frac{1}{\ell_t}$, and the probability of this award to be the maximal is $\frac{1}{\ell_t}$, and to be the second maximal is $\frac{n-t}{n}$. The second inequality is by the induction hypothesis. In this case, if $t=\frac{n}{2} $ and $\ell_t=2$, then both values of $\ell'$ give same lower bound.

Similarly,
\begin{eqnarray*}
	B_t^{\ell_t} &\geq&  \min_{1\leq \ell'\leq \ell_t }\left\{\frac{1}{\ell'} \cdot \left(\frac{t}{n},\frac{n-t}{n}\right) + \frac{\ell'-1}{\ell'}D_t^{\ell_t-1} \right\}\\
	&\geq&  \min_{1\leq \ell'\leq \ell_t }\left\{\frac{1}{\ell'} \cdot \left(\frac{t}{n},\frac{n-t}{n}\right) + \frac{\ell'-1}{\ell'}\left(\frac{n-t}{n(\ell_t-1)},\frac{t(n-t)}{n(n-1)}\right)\right\}\\  & =  &
	\frac{1}{2} \cdot \left(\frac{t}{n},\frac{n-t}{n}\right) + \frac{1}{2}\left(\frac{n-t}{n},\frac{t(n-t)}{n(n-1)}\right)\\  & = &  \left(\frac{1}{2},\frac{(n+t-1)(n-t)}{2n(n-1)}\right),
\end{eqnarray*} where the first inequality holds since the probability of receiving the award under competition is $\frac{1}{\ell_t}$, and the probability of this award to be the maximal is $\frac{1}{\ell_t}$, and to be the second maximal is $\frac{n-t}{n}$. The second inequality is by the induction hypothesis. In this case if $t=\frac{n}{2}$ and $\ell_t=2$, the unique minimum is reached at $\ell'=2$.

\item  \ronedit{For the case where $\ell_t=2$ and $\frac{t}{n} < \frac{1}{\ell_t}$.  According to $\sigma$, agent $i$ does not compete over any of the available awards. Thus, either no other agent selects an award and the winning probability of agent $i$ is at least $A_{t+1}^{\ell_t}$, or the number of active agents decreases by one. Thus,
	
	\begin{eqnarray*}
		A_t^{\ell_t} &\geq&  \min\left(C_t^{\ell_t-1},A_{t+1}^{\ell_t} \right)\\
		&\geq& \min\left(	\left(\frac{n-t}{n\ronedit{(\ell_t-1)}},\frac{t}{n\ronedit{(\ell_t-1)}}\right),\left( \frac{1}{\ell_t},\frac{1}{\ell_t}\right) \right)\\  
		& = &  \left(\frac{1}{\ell_t},\frac{1}{\ell_t}\right),
	\end{eqnarray*} 
	 where the first inequality holds since the first term lower bounds the case where there exists another agent that selects the maximal available award, and the second term lower bounds the case that no other agent selects the maximal available award. The second inequality is by the induction hypothesis.
	
	Similarly,
	\begin{eqnarray*}
		B_t^{\ell_t} &\geq&  \min\left(D_t^{\ell_t-1},\frac{t-1}{t+1}B_{t+1}^{\ell_t} + \frac{2}{t+1}A_{t+1}^{\ell_t}\right)\\
		&\geq& \min\left(	\left(\frac{n-t}{n(\ell_t-1)},\frac{t(n-t)}{n(n-1)(\ell_t-1)}\right),\frac{t-1}{t+1}\left(\frac{1}{\ell_t},\frac{(n-t-1)(n+t)}{n(n-1)\ell_t}\right)+\frac{2}{t+1}\left( \frac{1}{\ell_t},\frac{1}{\ell_t}\right) \right)\\  
		& = &  \left(\frac{1}{\ell_t},\frac{(n-t)(n+t-1)}{n(n-1)\ell_t}\right),
	\end{eqnarray*} where the first inequality holds since the first term lower bounds the case where there exists another agent that selects the maximal available award, and the second term lower bounds the case that no other agent selects the maximal available award, and thus at time $t+1$, both the maximal and second maximal awards will be available with probability  $\frac{2}{t+1}$. The second inequality is by the induction hypothesis.}

\item For  the case where $\ell_t>2$  and  $\frac{t}{n} > \frac{1}{\ell_t}$, according to $\sigma$, agent $i$ competes over the maximal award so far. Assuming that $\ell' $ agents (including $i$) compete over this award, we get that:
\begin{eqnarray*}
	A_t^{\ell_t} &\geq&  \min_{1\leq \ell'\leq \ell_t }\left\{\frac{1}{\ell'} \cdot \left(\frac{t}{n},\frac{t(n-t)}{n(n-1)}\right) + \frac{\ell'-1}{\ell'}C_t^{\ell_t-1} \right\}\\
	&\geq&  \min_{1\leq \ell'\leq \ell_t }\left\{\frac{1}{\ell'} \cdot \left(\frac{t}{n},\frac{t(n-t)}{n(n-1)}\right) + \frac{\ell'-1}{\ell'}\left(\frac{n-t}{n(\ell_t-1)},\frac{n^2+t^2-tn-n}{n(n-1)(\ell_t-1)}\right)\right\}\\  & =  &
	\frac{1}{\ell_t} \cdot \left(\frac{t}{n},\frac{t(n-t)}{n(n-1)}\right) + \frac{\ell_t-1}{\ell_t}\left(\frac{n-t}{n(\ell_t-1)},\frac{n^2+t^2-tn-n}{n(n-1)(\ell_t-1)}\right)\\  & = &  \left(\frac{1}{\ell_t},\frac{1}{\ell_t}\right),
\end{eqnarray*} where the first inequality holds since the probability of receiving the award under competition is $\frac{1}{\ell_t}$, and the probability of this award to be the maximal is $\frac{1}{\ell_t}$, and to be the second maximal is $\frac{t(n-t)}{n(n-1)}$. The second inequality is by the induction hypothesis.

Similarly,
\begin{eqnarray*}
	B_t^{\ell_t} &\geq&  \min_{1\leq \ell'\leq \ell_t }\left\{\frac{1}{\ell'} \cdot \left(\frac{t}{n},\frac{t(n-t)}{n(n-1)}\right) + \frac{\ell'-1}{\ell'}D_t^{\ell_t-1} \right\}\\
	&\geq&  \min_{1\leq \ell'\leq \ell_t }\left\{\frac{1}{\ell'} \cdot \left(\frac{t}{n},\frac{t(n-t)}{n(n-1)}\right) + \frac{\ell'-1}{\ell'}\left(\frac{n-t}{n(\ell_t-1)},\frac{n-t}{n(\ell_t-1)}\right)\right\}\\  & =  &
	\frac{1}{\ell_t} \cdot \left(\frac{t}{n},\frac{t(n-t)}{n(n-1)}\right) + \frac{\ell_t-1}{\ell_t}\left(\frac{n-t}{n(\ell_t-1)},\frac{n-t}{n(\ell_t-1)}\right)\\  & = &\left(\frac{1}{\ell_t},\frac{(n-t)(n+t-1)}{n(n-1)\ell_t}\right),
\end{eqnarray*} where the first inequality holds since the probability of receiving the award under competition is $\frac{1}{\ell_t}$, and the probability of this award to be the maximal is $\frac{1}{\ell_t}$, and to be the second maximal is $\frac{t(n-t)}{n(n-1)}$. The second inequality is by the induction hypothesis.

\item For the case where \ronedit{$\frac{t}{n} \leq \frac{1}{\ell_t}$ and $\ell_t>2$,}\ron{\uline{$\frac{t}{n} < \frac{1}{\ell_t}$ and $\ell_t\geq 2$, or $\frac{t}{n}=\frac{1}{\ell_t}$ and $\ell_t>2$}} according to $\sigma$, agent $i$ does not compete over any of the available awards. Thus, either no other agent selects an award and the winning probability of agent $i$ is at least $A_{t+1}^{\ell_t}$, or the number of active agents decreases by one. Thus,

\begin{eqnarray*}
	A_t^{\ell_t} &\geq&  \min\left(C_t^{\ell_t-1},A_{t+1}^{\ell_t} \right)\\
&\geq& \min\left(	\left(\frac{n-t}{n\ronedit{(\ell_t-1)}},\frac{n^2+t^2-tn-n}{n(n-1)\ronedit{(\ell_t-1)}}\right),\left( \frac{1}{\ell_t},\frac{1}{\ell_t}\right) \right)\\  
 & = &  \left(\frac{1}{\ell_t},\frac{1}{\ell_t}\right),
\end{eqnarray*} 
where the first inequality holds since the first term lower bounds the case where there exists another agent that selects the maximal available award, and the second term lower bounds the case that no other agent selects the maximal available award. The second inequality is by the induction hypothesis.

Similarly,
\begin{eqnarray*}
	B_t^{\ell_t} &\geq&  \min\left(D_t^{\ell_t-1},\frac{t-1}{t+1}B_{t+1}^{\ell_t} + \frac{2}{t+1}A_{t+1}^{\ell_t}\right)\\
	&\geq& \min\left(	\left(\frac{n-t}{n(\ell_t-1)},\frac{n-t}{n(\ell_t-1)}\right),\frac{t-1}{t+1}\left(\frac{1}{\ell_t},\frac{(n-t-1)(n+t)}{n(n-1)\ell_t}\right)+\frac{2}{t+1}\left( \frac{1}{\ell_t},\frac{1}{\ell_t}\right) \right)\\  
	& = &  \left(\frac{1}{\ell_t},\frac{(n-t)(n+t-1)}{n(n-1)\ell_t}\right),
\end{eqnarray*} where the first inequality holds since the first term lower bounds the case where there exists another agent that selects the maximal available award, and the second term lower bounds the case that no other agent selects the maximal available award, and thus at time $t+1$, both the maximal and second maximal awards will be available with probability  $\frac{2}{t+1}$. The second inequality is by the induction hypothesis.

\item If the maximal award so far is not available, then according to $\sigma$, agent $i$ does not select an award. Thus,
\begin{eqnarray*}
	C_t^{\ell_t}  
	&\geq&  \frac{1}{t+1}B_{t+1}^{\ell_t}+\frac{t}{t+1}C_{t+1}^{\ell_t}  
	\\ 
	& \geq & \frac{1}{t+1} \left(\frac{1}{\ell_t},\frac{(n-t-1)(n+t)}{n(n-1)\ell_t} \right) +\frac{t}{t+1} \left(\frac{n-t-1}{n\ell_t},\frac{n^2+(t+1)^2-(t+1)n-n}{n(n-1)\ell_t} \right)
	 \\ & = &  \left(\frac{n-t}{n\ell_t},\frac{n^2+t^2-tn-n}{n(n-1)\ell_t}\right),
\end{eqnarray*} where the first inequality holds since the award at time $t+1$ is the maximal so far with probability $\frac{1}{t+1}$. The second inequality is by the induction hypothesis.
It also holds that:
\begin{eqnarray*}
	D_t^{\ell_t} 
	&\geq&  \frac{1}{t+1}B_{t+1}^{\ell_t}+\frac{1}{t+1}C_{t+1}^{\ell_t}+\frac{t-1}{t+1}D_{t+1}^{\ell_t}  
	\\ 
	& \geq & \frac{1}{t+1} \left(\frac{1}{\ell_t},\frac{(n-t-1)(n+t)}{n(n-1)\ell_t} \right) +\frac{1}{t+1} \left(\frac{n-t-1}{n\ell_t},\frac{n^2+(t+1)^2-(t+1)n-n}{n(n-1)\ell_t} \right) \\
	& + &\frac{t-1}{t+1} \left(\frac{n-t-1}{n\ell_t},\frac{n-t-1}{n\ell_t} \right)
	\\ & = &  \left(\frac{n-t}{n\ell_t},\frac{n-t}{n\ell_t}\right),
\end{eqnarray*} 
where the first inequality holds since the award at time $t+1$ is the maximal or second maximal so far with probability $\frac{1}{t+1}$. The second inequality holds by the induction hypothesis.
\end{enumerate}
\qed

\end{document}